\RequirePackage{amsmath}
\documentclass{article}

\usepackage[letterpaper,left=3.2cm,right=3.2cm,top=2.5cm,bottom=2.5cm,includehead]{geometry}
\usepackage[latin1]{inputenc}
\usepackage[T1]{fontenc}
\usepackage{amsfonts}
\usepackage{graphicx}
\usepackage{subfigure}
	\usepackage{natbib}
\usepackage{ntheorem}

\newtheorem{theorem}{Theorem} 
\newtheorem{proposition}[theorem]   {Proposition}
\newtheorem{corollary}[theorem]{Corollary}
\newtheorem{lemma}[theorem]{Lemma}
\theorembodyfont{\upshape}
\newtheorem{remark}[theorem]{Remark}
\theoremstyle{nonumberplain}
\theoremheaderfont{\bfseries}
\theorembodyfont{\normalfont}
\theoremsymbol{\ensuremath{_\blacksquare}}
\newtheorem{proof}{Proof}
\qedsymbol{\ensuremath{_\blacksquare}}

\newcommand{\MM}{{\mathbb M}}
\newcommand{\PP}{{\mathbb P}}
\newcommand{\EE}{{\mathbb E}}
\newcommand{\FF}{{\mathbb F}}

\newcommand{\D}{{\cal D}}
\newcommand{\N}{{\cal N}}
\newcommand{\X}{{\cal X}}

\newcommand{\Pa}{{\cal P}}
\newcommand\ind{{\bf 1}}
\newcommand\uzeta{\underline\zeta}

\newcommand\nA{\hbox{\rm A}}
\newcommand\nC{\hbox{\rm C}}
\newcommand\nG{\hbox{\rm G}}
\newcommand\nT{\hbox{\rm T}}
\newcommand\kd{d_{\hbox{k}}}
\newcommand\chid{d_{\chi^2}}
\newcommand\hz{\kappa}

\newcommand{\figref}[1]{Fig.~\ref{#1}}

\title{An Entropy-Based Technique for Classifying Bacterial Chromosomes According to
Synonymous Codon Usage}
\author{%
Andrew Hart%
\thanks{Centro de Modelamiento Matem\'atico, UMI 2071 CNRS-UCHILE, Facultad de Ciencias
F\'isicas y Matem\'aticas, Universidad de Chile,
Casilla 170, Correo 3, Santiago, Chile \newline
E-mail: \texttt{ahart@dim.uchile.cl}, Tel.: +56\,2\,2978\,4600, FAX:
+56\,2\,2688\,3821}
\and
Servet Mart{\'\i}nez%
\thanks{Departamento de Ingenier{\'\i}a Matem\'atica and Centro de 
Modelamiento Matem\'atico, UMI 2071 CNRS-UCHILE, Facultad de Ciencias
F\'isicas y Matem\'aticas, Universidad de Chile,
Casilla 170, Correo 3, Santiago, Chile \newline
E-mail: \texttt{smartine@dim.uchile.cl}}
}%
 \date{13 October 2015}

\begin{document}

\maketitle

\begin{abstract}
We present a framework based on conditional entropy and the Dirichlet 
distribution for classifying chromosomes based on the degree to which 
they use synonymous codons uniformly or preferentially, that is, whether 
or not codons that code for an amino acid appear with the same relative 
frequency. Applying
the approach to a large collection of annotated bacterial chromosomes 
reveals three distinct groups of bacteria.
\end{abstract}

\medskip
\noindent {\small

}

\bigskip

\section{Introduction}
\label{sec:intro}

Living cells use the genetic code to translate triples of nucleic acids 
called codons into amino acids, the building blocks of proteins.  There 
are a total of 64 codons of which 61 code for 20 amino acids while the 
remaining 3 constitute translation stop signals. This many-to-one mapping 
between codons and amino acids means the genetic code is degenerate. As 
almost all amino acids are represented by two to six different codons, the 
code possesses intrinsic redundancy. this provides the cellular machinery 
with the ability to correctly manufacture protein products in the presence 
of certain kinds of transcription/translation/replication errors. However, 
the way in which codons are used to represent amino acids varies from gene 
to gene and from organism to organism. The mechanism by which variations 
in these usage patterns arise is not clearly understood, though a number 
of factors that can influence codon usage are known. these include 
mutational biases, translational selection pressures, GC content at the 
third codon site (GC3s) and gene size. The pattern in the way amino acids 
are represented by codons is called synonymous codon usage (SCU) and 
inequity in the distribution of codons which code for the same amino acid 
is referred to as codon usage bias (CUB).

\medskip

As noted above, patterns of SCU are specific to each organism and 
their study is of significant biological interest. Consequently, 
there is a substantial body of literature devoted to studying SCU 
in the genes of organisms. \citet{comeron&aguade1998} give a brief 
review of methods for measuring SCU, both general and species-specific. 
Two useful measures are the codon adaptation index \citep{sharp&lee1987b} 
and the `effective number of codons' \citep{wright1990}. A number of 
other measures have been subsequently posed such as relative codon 
adaptation \citep{fox&erill2010}, various modifications to the 
`effective number of codons' \citep{fuglsang2006,banerjee&etal2005} and 
the measure of gene expression $E(G)$ which was used 
by~\citet{karlin&etal2001} to characterize predicted highly 
expressed genes in four bacteria.

\medskip

Here, we develop a method for assessing codon usage bias based on concepts 
taken from information theory and apply it to a large set of bacterial 
chromosomes. As our aim is to explore and identify trends in global SCU in 
a large collection of organisms, our approach differs from that which is 
typically taken when studying SCU in a chromosome.  Rather than examining 
a chromosome gene by gene, we aggregate all the genes annotated for the 
chromosome together and compute total relative frequencies for codons and 
amino acids which are then plugged into our measure.

\medskip  

To set the scene more concretely, let $\{\nA,\nC,\nG,\nT\}$ be the 
alphabet of nucleotides in DNA. Being triplets of nucleotides, codons 
are elements of $\{\nA,\nC,\nG,\nT\}^3$ and $61$ of the $64$ possible 
codons code for amino acids. The remaining $3$ codons 
(TAA, TAG and TGA) indicate a STOP condition that is transcribed 
into messenger RNA and which instructs the ribosome to halt the 
translation of a sequence of codons into a polypeptide chain. 

\medskip

Each amino acid can be represented by any of its so-called synonymous 
codons, which are those codons that code for it.  For instance, the four 
codons GCA, GCG, GCT and GCC all code for the amino acid alanine (Ala). 
Thus the list of amino acids defines a partition 
$\alpha^*=\{A^*_d: d=1,..,20\}$ 
on the set of $61$ codons. Each amino acid is associated with a 
class $A^*_d$ which is the set of its synonymous codons. For example, 
if $A_1^*$ corresponds to Ala, then 
$A_1^*=\{\text{GCA}, \text{GCG}, \text{GCT}, \text{GCC}\}$. 
Further, the collection of $20$ amino acids can be organized 
according to the number of synonymous codons each has: there are 
$5$ amino acids with $4$ synonymous codons, $3$ with $6$ codons, 
$9$ with $2$ codons, $1$ with $3$ codons and $2$ with just $1$ codon.

\medskip

We seek a measure or statistic~$\Delta$ that captures the degree to 
which the set of amino acids are represented equally often by their 
synonymous codons.  If we suppose that there are $k_d$ codons that code 
for amino acid $A_d$, then $A_d$ will exhibit no CUB if each of 
its $k_d$ synonymous codons appears with relative frequency $1/k_d$. 
Extending this, it should be clear that a complete chromosome will 
exhibit no CUB provided that 
every codon~$C$ appears in the chromosome with a relative frequency 
equal to the reciprocal of the number of codons that code for the same 
amino acid as~$C$. In order for ~$\Delta$ to be useful, it should possess 
a number of desirable properties. Firstly, a complete absence of CUB, 
the ideal case just described, should be indicated by the reference 
value of~$0$. Secondly, the most extreme form of CUB where each amino 
acid is always represented by the same codon should correspond to the 
maximum value of~$\Delta$. 
Thirdly, larger values of~$\Delta$ should 
correspond to greater concentrations of the codon distribution on a 
smaller number of codons, with the extreme case being one codon per 
amino acid. 
Finally, it should be independent of the amino acid composition 
so that comparisons can be made between chromosomes and/or organisms.

\medskip

The following sections develop such a measure in a very general 
framework, replacing the set of codons by an arbitrary finite 
set~$I$ of symbols and the amino acid partition $\alpha^*$ by 
a general partition $\alpha$ of~$I$. Section \ref{sec2} begins by 
presenting information 
theoretic concepts such as partitions, entropy, conditional entropy and 
maximal entropy of probability measures when fixing a probability 
measure over some partition and defines a statistic~$\Delta$ for CUB. The
relationship of~$\Delta$ to probabilistic measures such as
the Kullback and $\chi^2$ distances is investigated.  It is also
compared to homozygosity, a well-known measure of genetic similarity.
Section \ref{sec3} deals with computing the mean entropy under a 
Dirichlet distribution assumption on the use of symbols belonging to the 
same atom. This enables the expected value of ~$\Delta$ 
to be computed directly from the relative frequencies of the 
amino acids. In Section \ref{sec4}, we obtain some large 
deviation bounds with respect to the Dirichlet distribution assumption. 

\medskip

The final section applies the proposed entropy-based approach to a set 
of $2535$ annotated bacterial chromosomes and we discover that the 
bacteria can be divided into three broad groups based on 
their SCU behavior. In the largest group, which contains $1587$ bacteria, 
amino acids are represented uniformly by their synonymous codons. 
The next largest group consists of $592$ bacteria that exhibit a very 
high degree of CUB, indicating an extreme preference for a 
small number of codons. the third and smallest group has $356$ bacteria 
whose SCU bias is moderate.

\section{Partitions, Entropy and Maximal Entropy} 
\label{sec2}

Let $I$ be a finite set and denote its cardinality by $|I|$.
A partition $\alpha=(A: A\in \alpha)$ of $I$ is such that its elements,
which we shall call atoms, are non-empty, disjoint and cover $I$, 
that is,
\begin{align*}
A \neq \emptyset \text{ for $A\in \alpha$} 
&& \text{ (non-emptiness);} \\
A\cap A' = \emptyset \text{ for $A\neq A'$, $A,A'\in \alpha$} 
&& \text{ (disjointness);} \\
I =\bigcup_{A\in \alpha} A 
&& \text{ (coverage).}  
\end{align*}
For an element $i\in I$ we use $A^{i,\alpha}\in \alpha$ to denote the 
unique atom of $\alpha$ that contains $i$.

\medskip

A partition $\beta=(B: B\in \beta)$ is said to be finer than $\alpha$ (or
$\alpha$ is coarser than $\beta$), written
$\beta\succeq \alpha$, if every atom of $\alpha$ is a union of atoms of 
$\beta$. Let $\beta\succeq \alpha$. 
For every $B\in \beta$ there is a unique atom 
$A\in \alpha$ containing $B$ and we denote this $A$ by $A^{B,\alpha}$.
For $A\in \alpha$, we define
$$
N^{\beta}(A)=|\{B\in \beta: B\subseteq A\}|,
$$
the number of atoms of $\beta$ contained in $A$.
So, for $B\in \beta$, $N^\beta(A^{B,\alpha})$ is the number of atoms
of $\beta$ which are in the atom of $\alpha$ containing $B$.
In the following, an increasing sequence of partitions should be understood
to mean that the sequence of partitions is increasing with respect to the order
$\preceq$.

\medskip

The finest partition is the discrete one $\D=\{\{i\}: i\in I\}$ 
while the coarsest is the trivial one $\N=\{I\}$. 
Let $\sigma(\alpha)$ denote the field generated by partition $\alpha$, that
is, the class of sets that are unions of atoms of $\alpha$.  
We denote the set of probability measures 
on $(I,\sigma(\alpha))$ by $\Pa(I,\sigma(\alpha))$. 

\medskip 

Once and for all, we fix two partitions $\alpha$ and $\beta$ 
satisfying $\beta\succ \alpha$. Also we fix a probability 
distribution $q\in \Pa(I,\sigma(\alpha))$. This distribution
can be represented by a vector
$q=(q_A: A\in \alpha)$ satisfying $q_A\ge 0$ for
$A\in \alpha$ and $\sum_{A\in \alpha}q_A=1$.
Its Shannon entropy $h(\alpha,q)$ is given by
\begin{equation}
\label{eq1}
h(\alpha,q)=-\sum_{A\in \alpha} q_A \log q_A.
\end{equation}	
A probability measure $r\in \Pa(I,\sigma(\beta))$ 
extends $q$ if $r|_{\sigma(\alpha)}=q$, that is if
$q_A=\sum_{B\in \beta: B\subseteq A}r_B$ 
for all $A\in \alpha$.
We write this extension relation as $r\succeq q$.
In this case we have the increasing property
$h(\alpha,q)\le h(\beta,r)$ and thus the conditional entropy is
\begin{equation}
\label{eqII3}
h_r(\beta \,|\, \alpha)=h(\beta,r)-h(\alpha,q)
=-\sum_{A\in \alpha} q_A \left(\sum_{B\in \beta: B\subseteq A}
\frac{r_B}{q_A}\log \frac{r_B}{q_A}\right).
\end{equation}
Denote the set of extensions of~$Q$ by
$$
\Pa(I,\sigma(\beta) \,|\, q)=\{r\in \Pa(I,\sigma(\beta)): r\succeq q\}.
$$
From now on, we shall assume that $r\in \Pa(I,\sigma(\beta) \,|\, q)$. 

Next, define $r^M\in \Pa(I,\sigma(\beta))$ by
$$
\forall B\in \beta: \quad r^M_B=N^\beta(A^{B,\alpha})^{-1}
q_{A^{B,\alpha}},
$$
that is, $r^M$ gives the same weight to the atoms of $\beta$
which are contained in the same atom of $\alpha$.
An easy computation shows that 
$r^{M}\in \Pa(I,\sigma(\beta) \,|\, q)$
and its entropy is
$$
h(\beta,r^M)=\sum_{A\in \alpha} q_A\log N^\beta(A)+h(\alpha,q).
$$
Since the function $-x\log x$ is concave on the
interval $[0,1]$, Jensen's inequality implies that $r^M$
maximizes the entropy over the set of probability measures that
extend $q$, that is,
$$
h(\beta,r^M)=\max\{h(\beta,r): 
r\in \Pa(I,\sigma(\beta) \,|\, q)\}.
$$

Once again letting $r\in \Pa(I,\sigma(\beta) \,|\, q)$, we define
\begin{equation}
\label{eqVI3}
\Delta^q(\beta,r)=h(\beta,r^{M})-h(\beta,r).
\end{equation}
In other words, $\Delta^q(\beta,r)$ is the difference between 
the maximal conditional entropy 
over the partition $\beta$ and the entropy of 
$r$ over $\beta$, both extending $q$. Then,
\begin{eqnarray}
\nonumber
\Delta^q(\beta,r) &=& \sum_{A\in \alpha} q_A\log N^\beta(A) +
h(\alpha,q) - h_r(\beta \,|\, \alpha) -h(\alpha,q) \\
\label{eqII1}
&=& \sum_{A\in \alpha} q_A \log N^\beta(A)- h_r(\beta \,|\, \alpha).
\end{eqnarray}

\medskip

Now, the statistic~$\Delta$ we shall use for assessing a genome's CUB 
will be
the special case of $\Delta^q(\beta,r)$ in which~$\beta$ is set to the
discrete partition of codons~$\D$, that is,
\begin{equation}
\label{eqn:delta.stat}
\Delta = \Delta^q(\D,r).
\end{equation}
As such, $\Delta$ is the difference between two entropies.
The first is the maximum conditional entropy over the discrete
partition~$\D$ of codons for probability measures extending the
amino acid composition~$q$ and the second is the entropy of the codon
distribution~$r$ which also extends~$q$.

\medskip

This statistic satisfies a number of properties which make it useful for
assessing CUB.  Being based on entropy, it is a natural quantity for
measuring departure from equal usage. It takes the value zero if and
only if~$r$ is constant on each atom of~$\alpha$, a configuration
concordant with unbiased SCU. The maximum value that~$\Delta$ can
take is $\sum_{A\in \alpha} q_A\log N^\beta(A)$. By the definition
of~$\Delta$, this maximum value corresponds to a conditional entropy of 
zero, that is, $h_r(\D\,|\,\alpha)=0$, and this means there is precisely 
one codon
representing each amino acid. Finally, fixing the distribution of
amino acids in the definition of~$\Delta$ has the effect of removing the
influence of amino acid composition so that~$\Delta$ can be used to 
compare the degree of CUB in coding sequences from different 
chromosomes/species.

\subsection{Relationship to Kullback and $\chi^2$ distances}

The quantity $\Delta^q(\beta,r)$ can be characterized in terms
of the Kullback `distance'. We recall that
if $\nu$ and $\mu$ are two probability measures on a measurable space 
with $\mu \ll \nu$ then the Kullback distance is 
$\kd(\mu,\nu)=\left( \int f \, \log f \, d\nu \right)^{1/2}$,
where $f=d\mu/\,d\nu$ is the Radon-Nikodym derivative of~$\mu$ with 
respect to~$\nu$.
 
\begin{proposition} 
For $r\in  \Pa(I,\sigma(\beta) \, | \, q)$, 
$\Delta^q(\beta,p)=\kd(r,r^M)^2$.
\end{proposition}

\begin{proof}
For $i\in B\in \beta$ we have
$$
\frac{d r}{d r^M}(i)=N^\beta(A^{B,\alpha})\frac{r_B}{q_A},
$$
so
$$
\kd(r,r^M)^2 = \sum_{A\in \alpha} 
\sum_{B\in \beta: B\subseteq A}
N^\beta(A^{B,\alpha})\frac{r_B}{q_A}
\log\left(N^\beta(A^{B,\alpha})\frac{r_B}{q_A}\right)
\frac{q_A}{N^\beta(A^{B,\alpha})}.
$$
Then
\begin{eqnarray*}
\kd(r,r^M)^2 &=& h(\alpha,q) - h(\beta,r) +
\sum_{A\in\alpha}\sum_{B\in\beta: B\subseteq A} r_B\log N^\beta(A) \\ 
&=& h(\alpha,q) - h(\beta,r) + \sum_{A\in \alpha}q_A\log
N^\beta(A)
=\Delta^q(\beta,r).
\end{eqnarray*}
\end{proof}

Let us recall the $\chi^2$ distance, which in the 
setting of $\mu \ll \nu$ and  $f=d\mu/ \, d\nu$
satisfies
$d_{\chi^2}(\mu,\nu)=\left( \int (f-1)^2 d\nu\right)^{1/2}$.
From a general formula \citep[see][]{gibbs&su2002}, we have 
\begin{equation}
\label{nva1}
\kd(\mu,\nu)\le d_{\chi^2}(\mu,\nu).
\end{equation}
In our case the $\chi^2$ distance is given by
\begin{eqnarray*}
d_{\chi^2}(r,r^M)^2 &=& \sum_{A\in \alpha} \sum_{B\in \beta: B\subseteq A}
\left(N^\beta(A)\frac{r_B}{q_A} -1 \right)^2
\frac{q_A}{N^\beta(A^{B,\alpha})} \\
&=&\sum_{A\in \alpha} \frac{q_A}{N^\beta(A)} 
\sum_{B\in \beta: B\subseteq A} 
\left( N^\beta(A)\frac{r_B}{q_A} -1\right)^2 \\
&=& \sum_{A\in \alpha} \frac{q_A}{N^\beta(A)} 
\left(\sum_{B\in \beta: B\subseteq A}
\left(N^\beta(A)\frac{r_B}{q_A}\right)^2-2N^\beta(A) + N^\beta(A) \right) 
\\
&=& \sum_{A\in \alpha} \frac{q_A}{N^\beta(A)}
\sum_{B\in \beta: B\subseteq A}  
\left(N^\beta(A)\frac{r_B}{q_A}\right)^2 - 1 \\
&=&\sum_{A\in \alpha} q_A N^\beta(A) 
\sum_{B\in \beta: B\subseteq A} \left(\frac{r_B}{q_A}\right)^2 - 1.
\end{eqnarray*}

The $\chi^2$ distance can be related to the homozygosities of the atoms
of~$\alpha$. Statistical estimators of homozygosity were used previously to
construct a measure of CUB called the ``effective number of codons'' 
\citep{wright1990}. The homozygosity of each atom $A\in\alpha$ is
$$
\hz_A(\beta) = \sum_{B\in\beta : B\subseteq A} 
\left(\frac{r_B}{q_A} \right)^2.
$$
Then,
$$
\chid(\mu,\nu)^2
=\sum_{A\in \alpha} q_A N^\beta(A) 
\sum_{B\in \beta: B\subseteq A} \left(\frac{r_B}{q_A}\right)^2 - 1 
=\sum_{A\in \alpha} q_A N^\beta(A) \hz_A(\beta) - 1.
$$
Consequently, from (\ref{nva1}) we get 
$\Delta(\beta,r) \le \sum_{A\in \alpha} q_A N^\beta(A) \hz_A(\beta) - 1$.

\section{Conditional entropy on Dirichlet distributions}
\label{sec3}

First, we write everything in terms of a probability measure 
$p\in \Pa(I,\sigma(\D))$. This measure is characterized as a vector
$p=(p_i: i\in I)$ which satisfies $p_i\ge 0$ for $i\in I$ and
$\sum_{i\in I}p_i=1$. Then, $p(J)=\sum_{i\in J}p_i$ for $J\subseteq I$
and $p(\{i\})=p_i$. 

\medskip

We will take $p\in \Pa(I,\sigma(\D) \, | \, q)$ to be an extension of $q$.
Let $r=p|_{\sigma(\beta)}$ be the
restriction of $p$ to $\sigma(\beta)$ so 
$r\in \Pa(I,\sigma(\beta) \, | \, q)$. We write
$h_p(\beta)=h_r(\beta)$, 
$h_{p}(\beta \,|\, \alpha)=h_r(\beta \,|\, \alpha)$ and   
$\Delta^q(\beta,p)=\Delta^q(\beta,r)$. 

\medskip

We would like to have an idea of how $\Delta^q(\beta,p)$ behaves
probabilistically. Toward this end, we shall place a probability law 
on $\Pa(I,\sigma(\D)\,|\, q)$ that captures our a priori ignorance 
about the measure~$p$ and compute $\Delta^q(\beta,p)$ when~$p$ is chosen 
according to this law. Since the first term of $\Delta^q(\beta,p)$ in 
(\ref{eqII1}) does not depend on $p$, we need only examine the behavior 
of the quantity $h_p(\beta \,|\, \alpha)$.

\medskip

To begin, note that any 
$p\in \Pa(I,\sigma(\D) \, | \, q)$ is an element of 
$\Pa(I,\sigma(\D))$ which satisfies $\sum_{i\in A}p_i=q_A$
for $A\in \alpha$. So $p$ is characterized by the set of 
probability vectors of the form
$$
\left(q_A^{-1}(p_i: i\in A): A\in \alpha\right).
$$
Here $q_A^{-1}(p_i: i\in A)=(p_i/q_A: i\in A)$ is a 
probability vector that 
takes values in $S_{|A|-1}$, the simplex of dimension $|A|-1$.

\medskip

We fix the distribution $\PP$ on $\Pa(I,\sigma(\D) \, | \, q)$ as a 
product of Dirichlet distributions with its support constrained to the 
set of extensions of~$q$. More precisely, the random vector 
$P=(P_i: i\in I)$ which takes values in $\Pa(I,\sigma(\D) \,|\, q)$ 
is such that
\begin{equation}
\label{eqI2}
X(A)=q_A^{-1}(P_i: i\in A) \sim \,
\text{Dirichlet}\,(\underbrace{1,..,1}_{|A| \, \text{times}}),
\end{equation}
and $(X(A): A\in \alpha)$ are independent random vectors.
We recall that the density for 
$(Y_1,..,Y_a)\sim \, \text{Dirichlet}\, (\underbrace{1,..,1}_{a\, 
\text{times}})$ is 
$f_{Y_1,..,Y_a}(y_1,..,y_a)=(a-1)!\ind((y_1,..,y_a)\in S_{a-1})$.  

\medskip

We can now compute the value $\EE\bigl(\Delta^q(\beta,P)\bigr)$
with respect to this distribution. 

\begin{theorem}
Fix a partition~$\alpha$ and a probability distribution~$q$ on~$\alpha$. 
Let $\beta \succeq \alpha$ and let $P$ be a random probability vector 
in $\Pa(I,\sigma(\D) \,|\, q)$ which is distributed according to 
(\ref{eqI2}). Then 
\begin{equation}
\label{eqII2}
\EE(h_P(\beta \,|\, \alpha)) 
=\sum_{A\in\alpha}q_A \xi_{|A|}-\sum_{A\in\alpha}
\frac{q_A}{|A|} 
\left(\sum_{B\in\beta: B\subseteq A}|B|\xi_{|B|}\right),
\end{equation}
where $\xi_n=\sum_{j=1}^n j^{-1}$ denotes the $n^{\rm th}$ harmonic
number.
\end{theorem}

\begin{proof}
To begin, we recall a standard construction of 
$(Y_j: j=1,..,a) \sim\,\text{Dirichlet}\,(c_1,..,c_a)$
for positive real numbers $c_j$, $j=1,..,a$ \citep[see for 
instance Section 2.2.1 in][]{bertoin2006}. We take
independent random variables $G_j,$ $j=1,..,a$ with $G_j$ distributed as 
Gamma$(c_j,1)$. Then, 
$\bigl(G_j\big/\sum_{i=1}^a G_i: j=1,..,a\bigr)\sim \,$ 
Dirichlet$\,(c_1,..,c_a)$ which is independent of 
$\sum_{i=1}^a G_i\sim\,$Gamma$(\sum_{j=1}^a c_j,1$).

\medskip

Now, if $(Y_j: j=1,..,a)\sim\,$ Dirichlet$\,(c_1,...,c_a)$ with
$c_1,..,c_a>0$, $(J_i: i=1,..,k)$ is a partition of $\{1,..,a\}$ and
$Z_i=\sum_{j\in J_i} Y_j$ for $i=1,..,k$, then it follows that
$(Z_i: i=1,..,k)\sim \,$ Dirichlet$\,\left(\sum_{j\in J_i} c_j:
i=1,..,k\right)$. 
As a special case of this, we have
\begin{equation}
\label{eqIII1}
(Z_i,1-Z_i)\sim \,\text{ Dirichlet }\left(\sum_{j\in J_i} c_j, 
\sum_{j\in \{1,..,a\}\setminus J_i} c_j\right).
\end{equation}
Recalling that $P(B)=\sum_{i\in B}P_i$, (\ref{eqII3}) yields
$$
h_P(\beta \,|\, \alpha)=-\sum_{A\in \alpha} q_A
\left(\sum_{B\in \beta: B\subseteq A}
\frac{P(B)}{q_A}\, \log \frac{P(B)}{q_A}\right)
$$
and combining distributions~(\ref{eqI2}) and~(\ref{eqIII1}) 
enables us to conclude that 
$$
\left(\frac{P(B)}{q_A},1-\frac{P(B)}{q_A}\right) \sim \,
\text{Dirichlet}(|B|, |A|-|B|)
$$ 
for $B\in \beta$, $B\subseteq A\in \alpha$. Then,
\begin{equation}
\label{eqII4}
\EE(h_P(\beta \,|\, \alpha))= -\sum_{A\in \alpha} q_A
\left(\sum_{B\in \beta: B\subseteq A} 
\EE(W_{B,A}\log W_{B,A})\right),
\end{equation}
where $W_{B,A}\sim \hbox{Dirichlet}\,(|B|,|A|-|B|)$.

\medskip

We need to compute $\EE(W \log W)$ for
$(W,1-W)\sim\,\text{Dirichlet}\,(l,m)$. Firstly,
$$
\EE\left(W \, \log W\right)=
\frac{(l+m-1)!}{(l-1)!\, (m-1)!}\int_0^1 x \, (\log x) \, 
x^{l-1}(1-x)^{m-1}dx.
$$
For integers $s,t\ge 0$, define
$$
\theta(s,t)=\int_0^1 x^s \, (1-x)^{t} \, \log x \, dx.
$$
This notation allows us to write
$$
\EE\left(W \log W\right)=
\frac{(l+m-1)!}{l-1)!\, (m-1)!}\theta(l,m-1).
$$
Thus, it is only necessary to compute $\theta(s,t)$, which will be done 
by successively integrating by parts and by making use of the formula
$$
\frac{d}{dx}A_s(x)=x^s \log x
\text{ where } A_s(x)=\frac{x^{s+1}\, \log x}{s+1}-\frac{x^{s+1}}{(s+1)^2}.
$$
To begin, we have $A_s(0)=0$ 
and $A_s(1)=-\frac{1}{(s+1)^2}$. Next for $t=0$ we have 
$\theta(s,0)=A_s(x)|_0^1=-\frac{1}{(s+1)^2}$.
Assume $t>0$. Since $A_s(x)(1-x)^t|_0^1=0$ we obtain
$$
\theta(s,t)=t \int_0^1 A_s(x)(1-x)^{t-1}dx=
\frac{t}{s+1}\theta(s+1,t-1)-
\frac{t}{(s+1)^2}\int_0^1 x^{s+1}(1-x)^{t-1}dx \,.
$$
If $t=1$ we obtain 
\begin{equation}
\label{eqI1}
\theta(s,1)=\frac{1}{s+1}\theta(s+1,0)-\frac{1}{(s+1)^2(s+2)}
=-\frac{(2s+3)}{(s+1)^2(s+2)^2} \,.
\end{equation}
If $t>1$ we find that
\begin{eqnarray*}
\theta(s,t)&=& \frac{t}{s+1}\theta(s+1,t-1)-
\frac{t}{(s+1)^2}\frac{(s+1)!(t-1)!}{(s+1+t)!}\\
&=&\frac{t}{s+1}\theta(s+1,t-1)-
\frac{1}{(s+1)}\frac{s!t!}{(s+1+t)!}.
\end{eqnarray*}
Then by iterating and by using (\ref{eqI1}), we  obtain
\begin{eqnarray*}
\theta(s,t)&=&\theta(s+t-1,1)\frac{s!\, t!}{(s+t-1)!}
-\frac{s!\, t!}{(s+t+1)!}\left(\sum_{j=1}^{t-1} \frac{1}{s+j}\right)\\
&=&-\frac{(2(s+t)+1)}{(s+t)(s+t+1)}\frac{s!\, t!}{(s+t+1)!}
-\frac{s!\, t!}{(s+t+1)!}\left(\sum_{j=1}^{t-1} \frac{1}{s+j}\right)\\
&=&-\frac{s!\, t!}{(s+t+1)!}\left(\frac{(2(s+t)+1)}{(s+t)(s+t+1)}
+\sum_{j=1}^{t-1} \frac{1}{s+j}\right) \\
&=&-\frac{s!\, t!}{(s+t+1)!} \left(\sum_{j=1}^{t+1} \frac{1}{s+j}\right) 
=-\frac{s!\, t!}{(s+t+1)!} \left( \xi_{s+t+1}-\xi_s \right).
\end{eqnarray*}
Hence, for $(W,1-W)\sim\,$Dirichlet$\,(l,m)$ we have  
\begin{eqnarray*}
\EE(W \, \log W)
&=& \frac{(l+m-1)!}{l-1)!\, (m-1)!} \theta(l,m-1) \\ 
&=& - \frac{(l+m-1)!}{(l-1)!\, (m-1)!} 
\cdot \frac{l!\, (m-1)!}{(l+m)!} \left( \xi_{l+m}-\xi_l \right) 
= -\frac{l}{l+m}\left(\xi_{l+m}-\xi_l\right).
\end{eqnarray*}
Setting $W=W_{B,A}$, $l=|B|$ and $m=|A|-|B|$, we obtain
$$
\EE(W_{B,A}\, \log W_{B,A}) = 
- \frac{|B|}{|A|}\left(\xi_{|A|}-\xi_{|B|}\right)
$$
and substituting this into (\ref{eqII4}) yields the result.
\end{proof}

\medskip

\noindent\textbf{Note}. The calculations in the above proof also enable 
$\EE\bigl(h_P(\beta \,|\, \alpha)\bigr)$ to be computed exactly when 
$X(A)\sim \,$Dirichlet $\,(\underbrace{c,..,c}\limits_{|A| \text{times}})$
where $c$ is a positive fixed integer which is the same for all 
$A\in \alpha$.
Here, we have fixed $c=1$ in the statement of the theorem because this 
causes all distributions~$p$ that extend $q$ to occur with the same 
probability.

\medskip

\begin{corollary}
\label{cor1}
We have
\begin{eqnarray*}
\EE\left(h(\beta, P)\right)
&=&
h(\alpha, q)+\sum_{A\in\alpha}q_A\xi_{|A|}-
\sum_{A\in\alpha}\frac{q_A}{|A|} 
\left(\sum_{B\in\beta: B\subseteq A}|B| \xi_{|B|}\right)\,;\\ 
\EE\left(\Delta^q(\beta,P)\right)
&=& \sum_{A\in \alpha} q_A\log N^\beta(A)-
\sum_{A\in\alpha}q_A\xi_{|A|}
+\sum_{A\in\alpha}\frac{q_A}{|A|}
\left(\sum_{B\in\beta: B\subseteq A} |B|\xi_{|B|}\right)\,.
\end{eqnarray*}
\end{corollary}

\section{A large deviation bound on maximality} 
\label{sec4}

Fix two partitions $\alpha$ and $\beta$ satisfying $\beta\succ \alpha$, together
with a probability measure $q\in \Pa(I,\sigma(\alpha))$.
For any observed~$P$, we can compute $\Delta^q(\beta, P)$. The corollary in the
preceding section gives the value of $\EE\left(\Delta^q(\beta,P)\right)$.
We would like to bound the expression
$$
\PP\left((\Delta^q(\beta,P)-\EE(\Delta^q(\beta,P)))\le -\lambda\right),
$$
to help us better understand the lower tail behavior of the distribution of
$\Delta^q(\beta,P)$ at the extreme near~$0$. This would provide an idea of how
typical a realization~$P$ is among all possible extensions of~$q$.

\medskip

An upper bound on this probability will be obtained by using the 
Azuma-Hoeffding large deviation inequality \citep{hoeffding1963,azuma1967}.
Since this inequality involves a filtration of $\sigma-$fields, it will be 
convenient to define a 
filtration where computations can be made and where the bound may be
sufficiently tight. Toward this end, we shall consider dyadic refinements, refinements that split a single atom. 

\medskip

More precisely, 
assume~$\gamma$ and~$\delta$ are two 
partitions that satisfy 
$\beta\succeq \delta\succeq \gamma\succeq \alpha$. We say that 
$\delta$ is a dyadic refinement of $\gamma$ if there is a unique atom 
$C^*\in \gamma$ that is split into two atoms $D^*,C^*\setminus D^*\in 
\delta$ while all
remaining atoms in $\gamma$ are also atoms of $\delta$, so
$\delta=\gamma\setminus \{C^*\}\cup \{D^*,C^*\setminus D^*\}$. 
Clearly, $|\delta|=|\gamma|+1$. From Corollary \ref{cor1} and by using
$A^*$ to denote the atom of $\alpha$ containing $C^*$, we find 
\begin{eqnarray}
\label{eqIV1}
\EE\left(\Delta^q(\delta,P)-\Delta^q(\gamma,P)\right)
&=&q_{A^*}\log\left(1\!+\!\frac{1}{N^\gamma(A^*)}\right) \\
\nonumber
&{}& +\frac{q_{A^*}}{|A^*|}
\left(|D^*|\xi_{|D^*|}+(|C^*|\!-\!|D^*|)\xi_{|C^*|-|D^*|}
-|C^*|\xi_{|C^*|}\right).
\end{eqnarray}

Now, a dyadic sequence of partitions from
$\alpha$ to $\beta$ is an increasing
sequence of partitions $\beta_{0,K}:=(\beta_k: k=0,..,K)$
from $\beta_0=\alpha$ to $\beta_K=\beta$ such that
$\beta_{k+1}$ is a dyadic refinement of $\beta_k$, that is,
$\beta_{k+1}=\beta_k \setminus 
\{B^k\}\cup \{\bar B^k,B^k\setminus \bar B^k\}$
for some $\bar B^k\subset B^k$ and all $k=0,..,K-1$. We 
let $A^k=A^{B^k,\alpha}$ denote the unique 
atom in $\alpha$ that contains the unique atom $B^k\in
\beta_k$ that is split in the refinement of 
$\beta_k$ to $\beta_{k+1}$. Note
that $|\beta_k|- |\alpha|=k$ for $k=0,..,K$. 
In particular, $|\beta|-|\alpha|=K$.
Also note that $B^0=A^0$ because $\beta_0=\alpha$.

\medskip

We shall associate the following sequence of real numbers 
$(\zeta_k: k=0,..,K-1)$ with a dyadic sequence of partitions 
$\beta_{0,K}$ from~$\alpha$
to~$\beta$:
\begin{equation}
\label{eqVII1}
\zeta_k=\max\{G_k, H_k+L_k\}\,, 
\end{equation}
where
\begin{eqnarray}
\nonumber
G_k &=& q_{A^k} 
\log \max\{ N^{\beta}(\bar B^k), N^{\beta}(B^k\setminus \bar B^k)\}
+\log2
+\min\{-q_{A^k}\log q_{A^k},e^{-1}\}\,;\\
\nonumber
H_k &=& q_{A^k}\,
\log N^\beta(B^k) +\min\{-q_{A^k}\,\log q_{A^k}, e^{-1}\}\,; \\
\label{eqVIII1}
L_k&=& \frac{q_{A^k}}{|A^k|}\,
\left(|\bar B^k|(\xi_{|B^k|}-\xi_{|\bar B^k|})+
|B^k\setminus \bar B^k|(\xi_{|B^k|}-\xi_{|B^k\setminus \bar B^k|})\right)\,. 
\end{eqnarray}

\begin{theorem}
\label{thm2}
Let $\beta\succ \alpha$. For all $\lambda>0$ we have
\begin{equation}
\label{azuma-hoeffding}
\PP(\Delta^q(\beta,P)\le \EE(\Delta^q(\beta,P))-\lambda)\le
\exp\left(-\frac{\lambda^2}{2\, \sum_{k=0}^{K-1}
\zeta_k^2}\right),
\end{equation}
where $\beta_{0,K}$ is a dyadic sequence of partitions from 
$\alpha$ to $\beta$ and $(\zeta_k: k=0,..,K-1)$ is given by 
(\ref{eqVII1}) and (\ref{eqVIII1}).

\medskip

Moreover, when the following condition holds,
\begin{equation}
\label{eqsufcond}
q_A\le \frac{1}{2} \text{ for all } A\in \alpha,
\end{equation}
we have $\zeta_k=G_k$, $k=0,..,K-1$.

\medskip

Finally, if 
$q_A\le e^{-1}$ for all $A\in \alpha$,
$\zeta_k$ becomes
\begin{equation}
\label{neq4}
\zeta_k = q_{A^k} \log \max\{N^\beta(\bar B^k), 
N^\beta(B^k\setminus\bar B^k)\} + \log2 -q_{A^k}\log q_{A^k}.
\end{equation}
\end{theorem}

\begin{proof}
First, associate the following set of 
random variables with a partition $\beta$:
\begin{equation}
\label{eqIV2}
(\X_B: B\in \beta) \, \text{ where } \, 
\X_B=q_{A^{B,\alpha}}^{-1}\; P(B).
\end{equation}
Let $\beta_{0,K}:=(\beta_k: k=0,..,K)$ be
a dyadic sequence of partitions from $\beta_0=\alpha$ to $\beta_K=\beta$.
Define a sequence of increasing 
$\sigma-$fields $\FF=(\FF_k: k=0,..,K)$ as follows:
\begin{equation}
\label{eqIV3}
\FF_k=\sigma(\X_B: B\in \beta_{k})  \hbox{ for } k=0,..,K.
\end{equation}
Note that $\FF_0$ is trivial because $\X_A=1$ for all $A\in \alpha$.

\medskip

Then, we can define a sequence of random variables $(\MM_k:  k=0,..,K)$ 
which is adapted to~$\FF$:
\begin{equation}
\label{eqV1}
\MM_K=\Delta^q(\beta,P)-\EE(\Delta^q(\beta,P))\,, \hbox{ and } \,
\forall \, k=0,...,K-1:\; \MM_k=\EE(\MM_K \, | \, \FF_k).
\end{equation}
It can be seen that $\MM_K$ is $\FF_K-$measurable because $\beta_K=\beta$, and
$(\MM_k: k=0,..,K)$ is a $\FF-$martingale. Observe that $\MM_0=0$ 
since $\FF_0$ is trivial.

\medskip

From (\ref{eqVI3}) and the second equality in Corollary \ref{cor1}, we have:
\begin{eqnarray}
\nonumber
\MM_K&=&
\sum_{B\in \beta}P(B)\log P(B)
-\sum_{A\in \alpha} q_A\log q_A+
\sum_{A\in\alpha}q_A \xi_{|A|} \\
\label{eqVII2}
&{}& \; -\sum_{A\in\alpha}\frac{q_A}{|A|}
\left(\sum_{B\in\beta: B\subseteq A} |B|\xi_{|B|}\right).
\end{eqnarray}
Now, for $k=0,..,K-1$ and all $B\in \beta$ 
satisfying $B\cap B^k=\emptyset$, we have
$$
\EE(P(B)\log P(B) \,|\, \FF_k)=\EE(P(B)\log P(B) \,|\, \FF_{k+1}).
$$
Hence
\begin{eqnarray}
\nonumber
\MM_{k+1}-\MM_k &=& \sum_{B\in \beta: B\subseteq B^k}
\EE\left(P(B)\log P(B) \,|\, \FF_{k+1}\right)
+\sum_{B\in \beta: B\subseteq B^k}
\EE\left(-P(B)\log P(B) \,|\, \FF_k\right) \\
\label{eqVII3}
&{}& \; +\frac{q_{A^k}}{|A^k|}
\left(|\bar B^k|(\xi_{|B^k|}-\xi_{|\bar B^k|})+
|B^k\setminus \bar B^k|(\xi_{|B^k|}-\xi_{|B^k\setminus \bar B^k|})\right).
\end{eqnarray}
Since the first term on the right-hand side is non-positive and 
the second and third terms are non-negative we get
$$
|\MM_{k+1}-\MM_k|\le \max\{-G'_k, H'_k+L_k\},
$$
where
\begin{eqnarray*}
G_k' &=& \sum_{B\in \beta: B\subseteq B^k}
\EE\left(P(B)\log P(B) \,|\, \FF_{k+1}\right)  \text{ and} \\
H_k' &=& -\sum_{B\in \beta: B\subseteq B^k}
\EE\left(P(B)\log P(B) \,|\, \FF_k\right)\,.
\end{eqnarray*}

We start by computing a bound on $-G'_k$. First, we 
shall obtain an upper bound for \\ 
$\EE(-P(B)\log P(B) \,|\, \FF_{k+1})$ 
by using Jensen's inequality:
$$
\EE(-P(B)\log P(B) \,|\, \FF_{k+1})\le -
\EE(P(B)\,|\, \FF_{k+1}) \log \EE(P(B)\,|\, \FF_{k+1}).
$$
Then
$$
\EE(P(B)\,|\, \FF_{k+1})=\frac{\EE(P(B))}{\EE(P(B^*))}P(B^*) 
\hbox{ where }  B^*\in \beta_{k+1},\; B\subseteq B^*.
$$
Furthermore, it follows from the 
definition of $P$ in~(\ref{eqI2}) that 
$q_A^{-1} \, |A| \, \EE(P_i)=1$ for all $i\in A$, so 
$$
\forall \, B\in \beta, B\subseteq A^k:\quad 
\EE(P(B))=q_{A^k}\frac{|B|}{|A|} \; \hbox{ and } \;
\EE(P(B) \,|\, \FF_{k+1})=\frac{|B|}{|B^*|}P(B^*),
$$
where $B\subseteq B^*\in\beta_{k+1}$. Hence,
$$
\EE(-P(B)\log P(B) \,|\, \FF_{k+1})\le - P(B^*) \frac{|B|}{|B^*|} 
\log \left(\frac{|B|}{|B^*|} \right) -\frac{|B|}{|B^*|}P(B^*)\log P(B^*)
$$
and therefore
\begin{eqnarray*}
&{}& \sum_{B\in \beta: B\subseteq B^k}\EE(-P(B)\log P(B) \,|\, \FF_{k+1})\\
&{}& \; \le -\sum_{B^*\in \beta_{k+1}: B^*\subseteq B^k} P(B^*) 
\sum_{B\in \beta: B\subseteq B^*} 
\frac{|B|}{|B^*|}\log \left(\frac{|B|}{|B^*|} \right) 
-\sum_{B^*\in \beta_{k+1}: B^*\subseteq B^k}P(B^*)\log P(B^*) \\
&{}& \; \le P(B^k) 
\log \max\{ N^{\beta}(\bar B^k), N^{\beta}(B^k\setminus \bar B^k)\} 
-\sum_{B^*\in \beta_{k+1}: B^*\subseteq B^k}\frac{P(B^*)}{P(B^k)} 
\log \left(\frac{P(B^*)}{P(B^k)}\right) \\ 
&{}& \quad -P(B^k)\log P(B^k) \\ 
&{}& \; \le q_{A^k}\, \sup\{\log N^{\beta}(B^*): B^*\in \beta_{k+1}\}
+\log N^{\beta_{k+1}}(B^k)-P(B^k)\log P(B^k).
\end{eqnarray*}

Now, the function $-x\log x$ is increasing on the interval $[0,e^{-1}]$ 
and so if $q_{A^k}\, \le e^{-1}$, it follows that 
$-P(B^k)\log P(B^k)\le -q_{A^k}\log q_{A^k}$. 
Otherwise $-P(B^k)\log P(B^k)\le e^{-1}$. 
Also, $N^{\beta_{k+1}}(B^k)=2$ as a consequence of 
the dyadic construction of the sequence of refinements $\beta_{0,K}$. 
Therefore, we obtain
\begin{eqnarray}
\nonumber
-G_k' &=& \sum_{B\in \beta: B\subseteq B^k}
\EE(-P(B)\log P(B) \,|\, \FF_{k+1}) 
\\
\nonumber
&\le & q_{A^k}\, \log \max\{ N^{\beta}(\bar B^k), 
N^{\beta}(B^k\setminus \bar B^k)\} +\log2
+\min\{-q_{A^k}\log q_{A^k}, e^{-1}\} \\
\label{eqVI4}
&=& G_k.
\end{eqnarray}
Analogously, we also have
\begin{eqnarray}
\nonumber
H_k' &\le& \sum_{B\in \beta: B\subseteq B^k}
\EE(-P(B)\log P(B) \,|\, \FF_k) \\
\nonumber
&\le & -\sum_{B^*\in \beta_k: B^*\subseteq B^k} P(B^*) 
\sum_{B\in \beta: B\subseteq B^*} 
\frac{|B|}{|B^*|}\log \left(\frac{|B|}{|B^*|} \right)
-\sum_{B^*\in \beta_k: B^*\subseteq B^k}P(B^*)\log P(B^*) \\
\nonumber
&\le & P(B^k)\log N^{\beta}(B^k)
-P(B^k)\log P(B^k) \\
\label{eqVI5}
&\le& q_{A^k}\, \log N^\beta(B^k) + 
\min\{-q_{A^k}\, \log q_{A^k}, e^{-1}\} = H_k.
\end{eqnarray}

By applying inequalities (\ref{eqVI4}) and (\ref{eqVI5}) to
(\ref{eqVII3}) and making use of  definitions (\ref{eqVII1}) and
(\ref{eqVIII1}), a straightforward argument gives
\begin{equation}
\label{eqVI2}
\forall k=0,...,K-1:\quad \sup\{|\MM_{k+1}-\MM_k|: k=0,..,K-1\}
\le \zeta_k.
\end{equation}

Then the Azuma-Hoeffding large deviation inequality
\citep[see for instance Lemma~11.2 in Section~11.1.4 of][]{waterman1995} 
gives the result. This inequality applies
to any martingale $(\MM_k:  k=0,..,K)$ with $\MM_0=0$
and $|\MM_{k+1}-\MM_k|\le \zeta_k$ a.s. for $k=0,..,K-1$:
$\PP(\MM_K\le -\lambda)\le \hbox{exp}\left(-\lambda^2 \big/ 
(2\sum_{k=0}^{K-1}\zeta_k^2)\right)$ for all $\lambda>0$. 

\medskip

Now, let us show that condition (\ref{eqsufcond}) implies 
$\zeta_k=G_k$, $k=0,..,K-1$. We will
require the following result in order to do this.

\begin{lemma}
\label{sintetico}
Let us consider the function for $t=1,..,s-1$ and $s=2,3,\ldots$:
\begin{equation}
\label{neq6}
\phi(s,t)= t(\xi_s-\xi_t) + (s-t)(\xi_s-\xi_{s-t})
= \sum_{u=t+1}^s\frac{t}{u} + \sum_{u=s-t+1}^s\frac{s-t}{u}.
\end{equation}
Then,
\begin{equation}
\label{neq7}
\phi(s,t) \leq s\log2 \hbox{ for all } t=1,..,s-1 \hbox{ and }s=2,3,\ldots.
\end{equation}
\end{lemma}

\begin{proof}  
First observe that $\phi(s,t)$ is symmetric in $t$ around $s/2$ in the
sense that $\phi(s,t)=\phi(s,s-t)$. If $s$ is even, then $\phi(s,t)$
takes its maximum value at $t=s/2$.
Defining $\bar\phi(s) = \max\{\phi(s,t): t=1,..,s-1\}$,
we have
$$
\bar\phi(s) = \phi(s,s/2) = s(\xi_s-\xi_{s/2})
= \sum_{u=s/2+1}^s\frac{1}{u} \leq s\int_{s/2}^s \; \frac{1}{x} \,dx
= s\bigl(\log s -\log(s/2)\bigr) = s\log2.
$$
On the other hand, if $s$ is odd, then $\phi(s,t)$
takes on its maximum value at
$t=(s\pm1)/2$. Fixing $l=(s-1)/2$ and $m=(s+1)/2=l+1$, we have
\begin{eqnarray*}
\bar\phi(s) &=& \phi(s,l) = \phi(s,m)
=l\sum_{u=l+1}^s\frac{1}{u} + m\sum_{u=m+1}^s\frac{1}{u}
= l\sum_{u=m+1}^s\frac{1}{u} +
m\sum_{u=m+1}^s\frac{1}{u} + \frac{l}{l+1} \\
&\le& s\sum_{u=m+1}^s\frac{1}{u} + \frac{s}{s+1}
= s\sum_{u=m+1}^{2m}\frac{1}{u} = \frac{s}{2m}\phi(2m,m),
\end{eqnarray*}
since $s+1=2m$. Above, it was proved that $\phi(2m,m)\le 2m\log2$, so
$\bar\phi(s)\le s\log2$ also if~$s$ is odd.
\end{proof}

\medskip

\noindent {\it Continuation of the Proof of Theorem \ref{thm2}}.
Note that $\zeta_k=G_k$ if
\begin{equation}
\label{GHL.ineq}
G_k\ge H_k+L_k.
\end{equation}

The function $L_k$, which  depends on $q_{A^k}$, $|B^k|$ and $|\bar B^k|$,
can be written in the form
$L_k=\frac{q_{A^k}}{|A^k|} \phi\bigl(|B^k|, |\bar B^k|\bigr)$.
From Lemma \ref{sintetico} we get the bound
$$
\max L_k = \frac{q_{A^k}}{|A^k|} \cdot
\bar\phi\bigl(|B^k|\bigr) \le \frac{|B^k|}{|A^k|}q_{A^k}\log2 \le
q_{A^k}\log2.
$$
Next, $G_k$ is a function of $q_{A^k}$, $B^k$ and $\bar B^k$.
However, only the first term of $G_k$ depends on $B^k$ and $\bar B^K$:
$q_{A^k}\log \max\{ N^\beta(\bar B^k), N^\beta(B^k\setminus \bar B^k)\}$.
This is clearly minimized when $N^\beta(B^k)$ is even and $N^\beta(\bar
B^k)=N^\beta(B^k)/2$ or $N^\beta(B^k)$ is odd and $N^\beta(\bar B^k)$ and
$N^\beta(B^k\setminus \bar B^k)$ differ by $1$. The second
term of $G_k$ is constant and the third only depends on $q_{A^k}$.  
Therefore, $G_k$ can  be minimized by selecting $\bar B^k$ so
that $N^\beta(\bar B^k) = \lceil N^\beta(B^k)/2\rceil$. This means
the minimum value of $G_k$ can be computed directly
from $q_{A^k}$ and $N^\beta(B^k)$, regardless of the
composition of $B^k$:
$$
\min G_k = q_{A^k} \log \lceil N^\beta(B^k)/2\rceil +\log2
+\min\{-q_{A^k}\log q_{A^k},e^{-1}\}.
$$
Then, a sufficient condition for (\ref{GHL.ineq}) to hold is
$\min G_k-H_k\ge \max L_k$, which is equivalent to
\begin{multline}
q_{A^k} \log\bigl( N^\beta(B^k)/2\bigr) + \log2
+ \min\{-q_{A^k}\log q_{A^k},e^{-1}\} -q_{A^k}\log N^\beta(B^k) \\
- \min\{-q_{A^k}\log q_{A^k},e^{-1}\} \ge q_{A^k}\log2.
\end{multline}
Rearranging and simplifying  reduces this condition to $q_{A^k}\le 1/2$.
Therefore, $\zeta_k=G_k$ for all $k=0,..,K-1$ whenever
$q_A\le 1/2$ for all $A\in\alpha$, which is exactly 
condition (\ref{eqsufcond}). 

\medskip

To show the last part of the Theorem, assume 
$q_A\le e^{-1}$ for all $A\in \alpha$. Relation
(\ref{neq4}) follows as a trivial consequence because 
$$
\zeta_k = G_k=q_{A^k} \log \max\{N^\beta(\bar B^k),
N^\beta(B^k\setminus\bar B^k)\} + \log2 -q_{A^k}\log q_{A^k}.
$$
\end{proof}

\begin{remark}
\label{satisfied}
In Section~\ref{sec5}, we examine a collection of 2535 bacterial
chromosomes downloaded from the NCBI ftp server.
Among these, we found that the largest value of
$\max\{q_A: A\in\alpha\}$ for any chromosome was $0.179855$, which
is well below $e^{-1}$.
Thus, condition~(\ref{eqsufcond}) is satisfied in all
the real-world bacterial chromosomes considered and likely holds in the
chromosomes of many other organisms and hence all the conclusions of
Theorem~\ref{thm2} are applicable.  
\end{remark}

\subsection*{\bf Computing the optimal bound}

When using (\ref{azuma-hoeffding}) to bound the tail probabilities of
$\Delta^q(\beta,P)$, it is desirable for $\sum_{k=0}^{K-1}\zeta_k^2$ to be as
small as possible.  

\medskip

From here on, we shall assume condition (\ref{eqsufcond}) holds, that is 
$q_A\le \frac{1}{2}$ for all $A\in \alpha$. Then
$\zeta_k=G_k$, $k=0,..,K-1$. In this case there is a straight 
forward strategy for selecting a $\beta_{0,K}$ which directly provides the 
minimum value of the bound
among all dyadically constructed sequences of partitions.  
This strategy is presented in the following theorem.
Clearly, the  $\zeta_k$'s depend on
the choice of the sequence of partitions $\beta_{0,K}$. As can be seen from
the theorem and its proof, there are numerous such sequences of partitions that
achieve the minimum bound and it is only necessary to select one of these.

\begin{theorem}
\label{thm3}
Fix two partitions~$\alpha$ and~$\beta$ such that $\beta\succ\alpha$. 
The following construction of $\beta_{0,K}$ minimizes 
$\sum_{k=0}^{K-1}\zeta_k^2$, and hence the bound on the right-hand 
side of (\ref{azuma-hoeffding}) among all
possible dyadically constructed sequences of partitions 
from~$\alpha$ to~$\beta$:

\begin{description}
\item[Step $0$.] 
Select any atom $A\in\alpha$ satisfying $N^\beta(A)>1$.  This
will be $B^0$.  Then choose $\bar B^0$ to be any subset of $B^0$ that is as
close to half the size of $B^0$ as possible, that is, $N^\beta(\bar B^0) =
\lceil N^\beta(B^0)/2\rceil$.
Use $B^0$ and $\bar B^0$ to construct $\beta_1$.
\item[Step $k$.]
For $k=1,..,K-1$,
choose an atom $B^k$in $\beta_k$ and then choose a subset~$\bar B^k$
of~$B^k$ such that $N^\beta(\bar B^k) = \lceil N^\beta(B^k)/2\rceil$. 
Then use $B^k$ and $\bar B^k$ to obtain
$\beta_{k+1}=\beta_k \setminus 
\{B^k\}\cup \{\bar B^k,B^k\setminus \bar B^k\}$.
\end{description}
\end{theorem}

\begin{proof}
Observe that the dyadic construction of the sequence of partitions
$\beta_{0,K}$ gives rise to the graph of a forest of binary trees 
with precisely one tree rooted at each atom of~$\alpha$. Every leaf of the 
$|\alpha|$ trees corresponds to an atom in~$\beta$. 
Every internal node in the graph has exactly two child nodes and 
corresponds to a~$B^k$ whose manner of splitting determines the 
associated~$\zeta_k$.  Henceforth, each node will be
identified with its $B^k$ and we shall use the terms node and
atom interchangeably. The first node to be split appears at level~$0$ 
in the graph while all the leaves are at level~$K$.  Each level 
corresponds to a partition in $\beta_{0,k}$: 
level~$k$ corresponds to $\beta_k$.  Finally, due
to the dyadic construction of the sequence of partitions, exactly 
one internal node appears at each level of the graph.

\medskip

As defined previously, each new level $k+1$ ($k=0,..,K-1$) is formed by 
taking a node~$B^k$ at the preceding level, which is
either~$A^k$ or has~$A^k$ as an ancestor, and splitting it into 
two child nodes, $\bar B^k$ and $B^k\setminus \bar B^k$. The nature of 
this split affects $\zeta_k$. Note that  
the root node~$A^k$ of the tree containing $B^k$ 
also plays a role in determining the value of $\zeta_k$. Each root
node~$A\in\alpha$ fixes the constant probability $q_A$ used to 
calculate all the $\zeta_k$'s associated with internal nodes in 
the tree rooted at~$A$.    
Hence, the effect of the probability distribution 
$q=(q_A: A\in\alpha)$ on the $\zeta_k$'s is not influenced 
in any way by the choice of $\beta_{0,K})$.

\medskip

Consequently, we can fix a sequence of partitions $\beta_{0,K}$ which
has been dyadically constructed as described in the preceding section, draw 
the associated graph and 
examine the $\zeta_k$'s corresponding to the nodes in a
particular tree independently of all the others.
So consider the tree rooted at some atom~$A\in\alpha$.  Define
$K(A)=\{k=0,..,K-1: B^k\subseteq A\}$ to be the set of indices at which nodes
belonging to the tree rooted at~$A$ are split.
Let $B^k$ for some $k\in K(A)$ be an internal node of this tree.
It was seen earlier that $\zeta_k$ is minimized by splitting~$B^k$ into~$\bar
B^k$ and $B^k\setminus \bar B^k$ such that $N^\beta(\bar B^k) = \lceil
N^\beta(B^k)/2\rceil$. Hence the size of~$B^k$ determines the minimum value 
of $\zeta_k$.  The actual composition of $\bar B^k$ and $B^k\setminus \bar B^k$ 
has no effect on $\zeta_k$. For the sake of brevity, we shall say that a node~$B^k$ is split in
half if $N^\beta(\bar B^k) = \lceil N^\beta(B^k)/2\rceil$.

\medskip

Now, it would seem that the obvious strategy for minimizing
$\sum_{k=0}^{K-1}\zeta_k^2$ would be to take each atom $A\in\alpha$ in turn and
recursively split it in half until all the leaf nodes are atoms of~$\beta$. 
This is the method described in the statement of the theorem.
Let $\uzeta_k$, $k\in K(A)$ be the values of the $\zeta_k$'s corresponding to nodes
in the tree rooted at~$A$ constructed using this recursive halving procedure:
$$
\uzeta_k = \min G_k = q_{A^k}\log\lceil N^\beta(B^k)/2\rceil +w_{A^k},
$$
where $w_A =\log2-q_A\log q_A$ for $A\in\alpha$.

\medskip
 
We will  show that this method does indeed minimize
$\sum_{k\in K(A)}\zeta_k^2$. Let $k^*$ be the first element in $K(A)$, that
is, where the root is split.  Then, $B^{k^*}=A$. $B^k$ and assume it splits into 
$B^k$ and $B^{k'}=A\setminus B^k$.  Suppose that~$A$ is not split in half, but $B^k$ and
$B^{k'}$ are subsequently split in half.  Let $s=N^\beta(A)$ and
$t=N^\beta(B^k)$. Without loss of generality, we assume that $B^k$ is the
larger of $A$'s two child nodes so that $t> \lceil s/2\rceil$. Then, we have
\begin{align*}
\zeta_{k^*} &= q_A\log t + w_A, \\ 
\zeta_k &= q_A\log\lceil t/2\rceil + w_A \text{ and} \\ 
\zeta_{k'} &= q_A\log\lceil(s-t)/2\rceil +w_A.
\end{align*}
Now, $\zeta_{k^*}$ only depends on~$A$, which is fixed by~$\alpha$, and $B^k$.
Since~$A$ is not split in half, $\zeta_{k^*}$ cannot attain its global
minimum value of $\uzeta_{k^*}=q_A\log\lceil s/2\rceil+w_A$.
In contrast, $\zeta_k$ (respectively $\zeta_{k'}$) depends on
how~$A$ was split in addition to how $B^k$ (respectively $B^{k'}$)
was split.  We have $\zeta_k>q_A\log\bigl\lceil\lceil s/2\rceil/2\bigr\rceil +
w_A=\uzeta_k$ and $\zeta_{k'}<q_A\log\bigl\lceil\lfloor s/2\rfloor/2\bigr\rceil
+ w_A=\uzeta_{k'}$.
A straightforward computation shows that any suboptimal split of~$A$ into atoms
of size $t$ and $s-t$ results in $\zeta_k^2+\zeta_{k'}^2>\uzeta_k^2+\uzeta_{k'}^2$.
Thus, even though $B^k$ and $B^{k'}$ are split in half, we will always have
\begin{equation}
\label{eq3zetabound}
\zeta_{k^*}^2+\zeta_k^2+\zeta_{k'}^2 
> \uzeta_{k^*}^2+\uzeta_k^2+\uzeta_{k'}^2.
\end{equation}

Once again, suppose that~$A$ is not split in half.   Now, if either 
$B^k$ or $B^{k'}$ is not split in half, then $\zeta_k$ and $\zeta_{k'}$ will be
at least as large as they were in the scenario described in the preceding
paragraph and thus $\zeta_k^2+\zeta_{k'}^2>\uzeta_k^2+\uzeta_{k'}^2$. Hence (\ref{eq3zetabound}) will continue to hold.
 
\medskip

Therefore, regardless of how the child nodes of~$A$ are split,
$\zeta_{k^*}^2+\zeta_k^2+\zeta_{k'}^2$ will always be greater than
$\uzeta_{k^*}^2+\uzeta_k^2+\uzeta_{k'}^2$ if~$A$ is split in half.
This argument is valid for all internal nodes of the tree, not just the root
node, and can be recursively applied to the tree in order to show that the
halving strategy will minimize $\sum_{k\in K(A)}\zeta_k^2$.  Since each tree is
 constructed independently of all the others, applying the strategy separately
 to each tree in the forest will minimize $\sum_{k=0}^{K-1}\zeta_k^2$.  This
 minimum value is $\sum_{k=0}^{K-1}\uzeta_k^2$.
   \end{proof}

\noindent\textbf{Note}. Single-node trees in the forest do not make any
contribution at all to $\sum_{k=0}^{K-1}\zeta_k^2$ as they have no internal nodes.

 \subsection*{\bf Large deviations for 
coding sequences}
	
Now, we can provide a more explicit calculation of
$\sum_{k=0}^{K-1}\uzeta_k^2$ for the coding regions of a genome.
Let~$\alpha=\{\text{Ala}, \text{Arg}, \ldots, \text{Val}\}$ be the set of 20
amino acids and $\D$ be the set of 61 non-STOP codons that code for the amino acids according to the
	standard genetic code.
		First of all, by setting $\beta=\D$, we can obtain
		Theorem~\ref{thm2} as it applies to the statistic~$\Delta$. 
In addition, recalling that $\max\{q_A: A\in\alpha\}<e^{-1}$ in bacterial
chromosomes, $\zeta_k=G_k$ for $k=0,..,K-1$ and $G_k$ simplifies to
$$
G_k = G\left(A^k, \max\{ |\bar B^k|, |B^k\setminus \bar B^k|\} \right),
$$
where
$$
G(A, t) = q_A\log t+\log2 -q_A\log q_A. 
$$
Note that $G(A,t)>\log2$ for all $A\in\alpha$ and $t\ge1$.

The set of amino acids can be grouped into 5 families according to how
many codons code for each amino acid.  Let $\alpha^{(l)}$ be the set of amino
acids which are coded for by~$l$ codons.  For example,
ATG is the only codon that codes for methionine (Met) while TGG is the
only one that codes for Tryptophan (Trp) and hence Met and Trp belong to
$\alpha^{(1)}$. The five families are:
 \begin{align*}
 \alpha^{(1)} &= \{\text{Met},\text{Trp}\} \\
 \alpha^{(2)} &= \{\text{Asn}, \text{Asp}, \text{Cys}, \text{Gln}, \text{Glu},
 \text{His}, \text{Lys}, \text{Phe}, \text{Tyr}\} \\
 \alpha^{(3)} &= \{\text{Ile}\} \\
 \alpha^{(4)} &= \{\text{Ala}, \text{Gly}, \text{pro}, \text{Thr},
 \text{Val}\}
 \\
 \alpha^{(6)} &= \{\text{Arg}, \text{Leu}, \text{Ser}\} \\
 \end{align*}

We can now present the following result which is a corollary of
Theorem~\ref{thm2}.

	\begin{corollary}
	\label{cor4}
	Let~$\alpha$ be the set of amino acids and~$\D$ the set of $61$ codons that
	code for amino acids according to the standard genetic code.  Assume the
	distribution $q=(q_A: A\in\alpha)$ satisfies $q_A<e^{-1}$ for all
	$A\in\alpha$.
	Then, for all $\lambda>0$, we have
	\begin{equation}
	\label{cotaX}
\PP\bigl(\Delta\le \EE(\Delta)-\lambda\bigr)
\le \exp\left( -\frac{\lambda^2}{2z(q)}\right),
\end{equation}
where
\begin{eqnarray}
\nonumber
z(q) = \sum_{k=0}^{K-1}\uzeta_k^2 
&=& \sum_{A\in\alpha^{(2)}\cup\alpha^{(3)}}G(A, 1)^2 +
2\sum_{A\in \alpha^{(4)}\cup\alpha^{(6)}}G(A, 1)^2 + \sum_{A\in
\alpha^{(3)}\cup\alpha^{(4)}}G(A,2)^2 \\
\label{eqGenomeZetaSS}
&&\quad + 2\sum_{A\in\alpha^{(6)}}G(A,2)^2 + \sum_{A\in\alpha^{(6)}}G(A,3)^2.
\end{eqnarray}
		\end{corollary}
	
	\begin{proof}
	Setting $\beta=\D$ in Theorem~\ref{thm2} causes $\Delta_q(\beta,P)$ to
	become~$\Delta$ in (\ref{azuma-hoeffding}). Next, use the
	recursive halving procedure given in Theorem~\ref{thm3} to construct a sequence
	of partitions $\beta_{0,K}$ that minimizes $\sum_{k=0}^{K-1}\zeta_k^2$.
Then, to complete the proof, we must show that $z(q)$ takes the form given
in~(\ref{eqGenomeZetaSS}).
	
\medskip

The graph associated with $\beta_{0,K}$ is a forest containing $20$
trees and has $K=|\D|-|\alpha|=61-20=41$ internal nodes and $61$ leaves.
	The trees rooted at Met and Trp, do not participate in determining the bound
	since they have no internal nodes.
	
	\medskip
	
	All the trees rooted at atoms in~$\alpha^{(2)}$ consist of a root node of
	size~$2$ and two leaf nodes.  Therefore, each $\zeta_k$ associated with
	$A\in\alpha^{(2)}$  will be
	$G(A, 1)=q_A\log1 + \log2-q_A\log q_A = \log2-q_A\log q_A$.  Consequently,
	these $\zeta_k$'s contribute
the following nine terms to $z(q)$:
$z_2(q)=\sum_{A\in\alpha^{(2)}}G(A, 1)^2$.

\medskip

Similarly, the tree rooted at Ile has two internal nodes (one of
size~$3$ and one of size~$2$) and $3$ leaf nodes.  Thus, Ile contributes $2$ terms to $z(q)$:
$z_3(q)=G(\text{Ile}, 1)^2+G(\text{Ile}, 2)^2$.
	
	\medskip
	
	Next, the five trees rooted at atoms in $\alpha^{(4)}$ have three internal
	nodes (one of size $4$ and two of size $2$) and 4 leaves.  Each contributes three terms to
	$z(q)$: $z_4(q)=\sum_{A\in\alpha^{(4)}}\bigl(G(A,2)^2+2G(A,1)^2\bigr)$.
	
	\medskip
	
	Finally, each of the trees rooted at an $A\in\alpha^{(6)}$ has five
	internal nodes, one of size~$6$, two of size~$3$ and two of size~$2$.  The sum
	of the squares of the $\zeta_k$'s associated with the five nodes of the tree
	rooted at~$A$ is $G(A,3)^2 + 2G(A,2)^2 + 2G(A,1)^2$.  Like $\alpha^{(4)}$,
	$\alpha^{(6)}$ accounts for fifteen terms of $z(q)$:
$z_6(q)=\sum_{A\in\alpha^{(6)}} \bigl(G(A,3)^2+2G(A,2)^2+2G(A,1)^2\bigr)$.

\medskip
	
	The proof is completed by rearranging $z(q) = z_2(q)+z_3(q)+z_4(q)+z_6(q)$. 
	\end{proof}

\begin{remark} 
\label{rembound} 
The corollary gives a bound on 
$$ 
\PP(\Delta\le \EE(\Delta)-\lambda).
$$ 
Obviously, the bound $\exp\left(-\lambda^2/2\,z(q)\right)$ only makes sense 
for $\lambda\in [0,\EE(\Delta)]$. Now, since $z(q)$ is the sum of $41$ terms 
of the form $G(A,t)>\log2$, we have $z(q)>41\cdot\log2$. On the other hand, 
$\EE(\Delta)\le 0.2773$ for the set of $2535$ chromosomes we have examined. 
Thus, the bound on the right-hand side of~(\ref{cotaX}) will be greater than or 
equal to
$$
\exp\left( -\frac{0.2773^2}{2\cdot 41\cdot\log2}\right) \approx 0.998648.
$$
Unfortunately, this means that the bound is not sufficiently 
tight for practical use.
\end{remark}

\begin{figure}[ht]
\centering
\subfigure[Density estimated from 2535 bacterial
chromosomes.]{\includegraphics[width=0.99\textwidth]{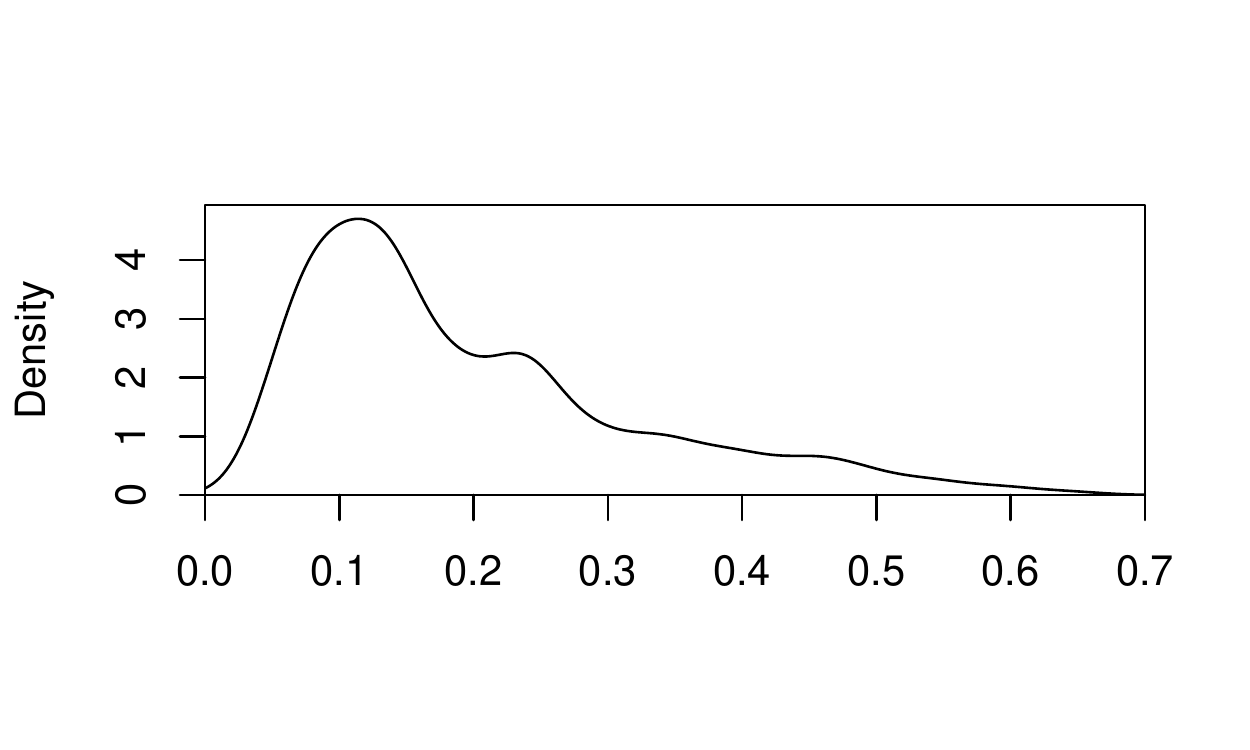}
\label{fig:delta.a}}

\subfigure[Density estimated from $10^4$ random
simulations.]{\includegraphics[width=0.99\textwidth]{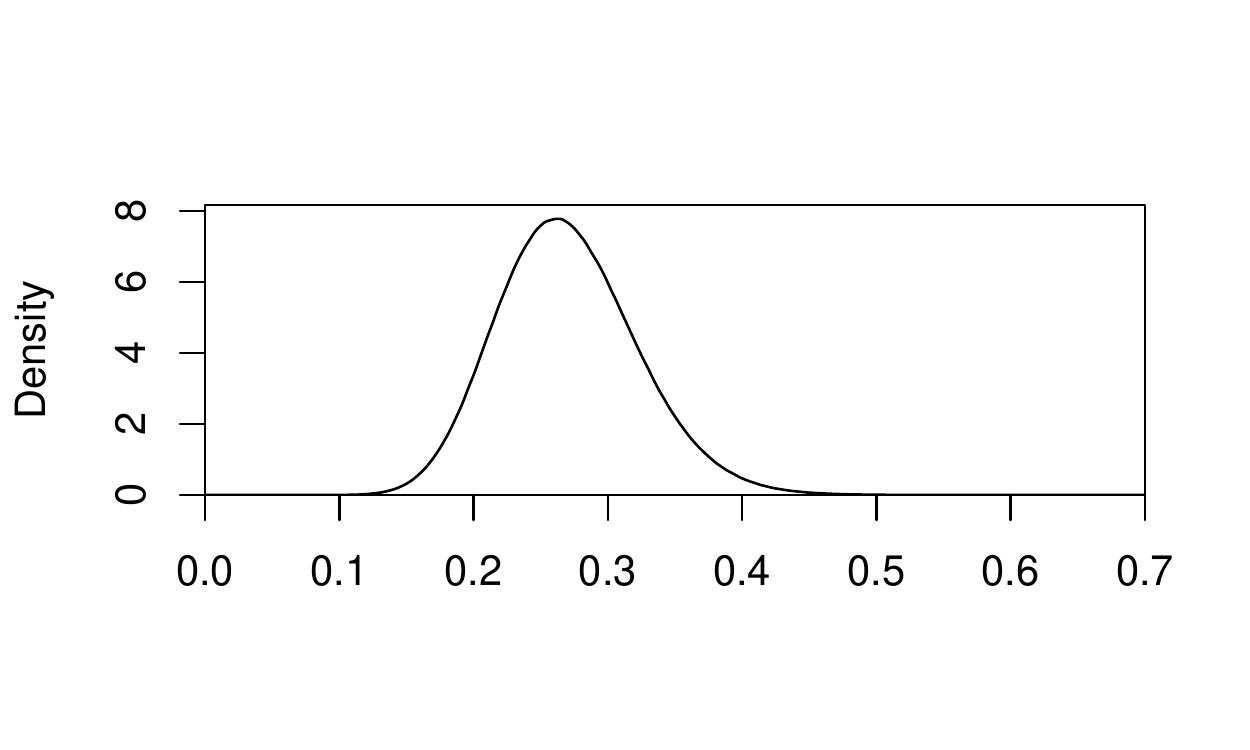}
\label{fig:delta.b}}
\caption{Distribution of $\Delta$ in real and randomly simulated data.}
\label{fig:delta}
\end{figure}

\section{Application and Comments}
\label{sec5}

\subsection*{\bf The data}

We downloaded a large set of 2585 bacterial DNA sequences
from the NCBI ftp server.  All of the sequences were marked as `complete
genome' or `nearly complete genome', so constitute chromosomes 
and not plasmids. Chromosomes which were either lacking
annotation data or which had fewer than 200 coding sequences, of 
which there were 8 and 45 respectively, were filtered out.
this left a set of 2535 chromosomes. Next, the codon distribution
($r$) was estimated from the relative frequencies of the codons 
for each chromosome by adding 
up the counts of codons in all the genes annotated in the 
GenBank (\texttt{.gbk}) file and rescaling the results to sum to one. 
For example, the numbers of AAA codons appearing in each gene were 
added together and then divided by the total number of codons contained 
in all the
genes in the chromosome. The corresponding amino acid distribution ($q$) 
was computed by summing up the relative frequencies of the synonymous 
codons for each amino acid.
Finally, the relative frequencies of the codons and amino acids were used 
to compute the ~$\Delta$ statistic for each chromosome in accordance
with (\ref{eqn:delta.stat}) and (\ref{eqII1}).

\medskip

In order to compare the $\Delta$'s obtain from real-world chromosome data
with the behavior of the~$\Delta$ statistic over the
complete range of codon distributions, we sampled 
$10^7$ codon distributions
uniformly and then computed the value of~$\Delta$ for each. 
Such a large sample
was utilized in order to capture sufficient detail at the extremes of the
distribution. A codon distribution can be sampled uniformly by simulating a
probability vector of length $61$ from a 
Dirichlet$(\underbrace{1,..,1}_{\text{\rm $61$ times}})$ distribution.  
Each of the $61$ components corresponds to a  single non-STOP codon.  

\medskip

\figref{fig:delta} shows estimates of the distribution of~$\Delta$ for
(a) the collection of 2535 chromosomes and (b) the simulated codon
distributions.  The range of~$\Delta$ 
is similar for both the set of chromosomes and the simulated data. 
However, the $\Delta$'s for the chromosomes are skewed towards the no CUB
end of the spectrum with fatter tails compared to the $\Delta$'s
based on uniformly sampled codon distributions, which present a symmetric aspect.
This suggests that the codon distributions for bacterial chromosomes 
are concentrated in a particular
region in the space of all possible codon distributions.

\medskip

A principal component analysis (PCA) of the codon distributions lends 
further evidence to support this conjecture. 
The first principal component is by far the most important, 
accounting for approximately $75$\% of variability in the 
codon distributions, while the first $10$ principal components 
explain about $95$\% of the total variability. This means that the codon 
distributions are essentially contained within a $10$-dimensional 
region inside the $61$-dimensional set of all theoretically 
possible codon distributions. Performing PCA on the amino acid
distributions yields a similar picture:
$80$\% of the variability is explained by the first principal 
component while the first $10$ components explain $98$\% of the total 
variability in the amino acid distributions.  

\medskip

Next, by examining where the~$\Delta$ computed for each bacterial chromosome 
lies in the set of all codon distributions that extend the amino acid 
distribution for that chromosome, an interesting phenomenon can be revealed. For 
each chromosome, we estimated $\Delta^*$, which is defined as follows: 
$\Delta^*$ is the proportion of all codon distributions that are compatible with 
the chromosome's amino acid distribution and that would give rise to a~$\Delta$ 
smaller than the chromosome's empirically estimated~$\Delta$. 

\medskip

In other words, we are interested  in the probability of observing a chromosome 
with less CUB than the chromosome under consideration and this probability is 
calculated naively by giving each possible codon distribution the same chance of 
being observed. Unfortunately, the bounds based on large deviations developed in 
Corollary~\ref{cor4} are not sharp enough to serve here (see 
Remark~\ref{rembound}). So instead, the probability for each chromosome  was 
estimated by carrying out a series of Monte Carlo simulations in which $10^4$ 
codon distributions were sampled uniformly at random subject to having the same 
amino acid distribution as the chromosome.

\medskip

\begin{figure}[ht]
\centering
\subfigure[Density for 2535 bacterial
chromosomes.]{\includegraphics[width=0.99\textwidth]
{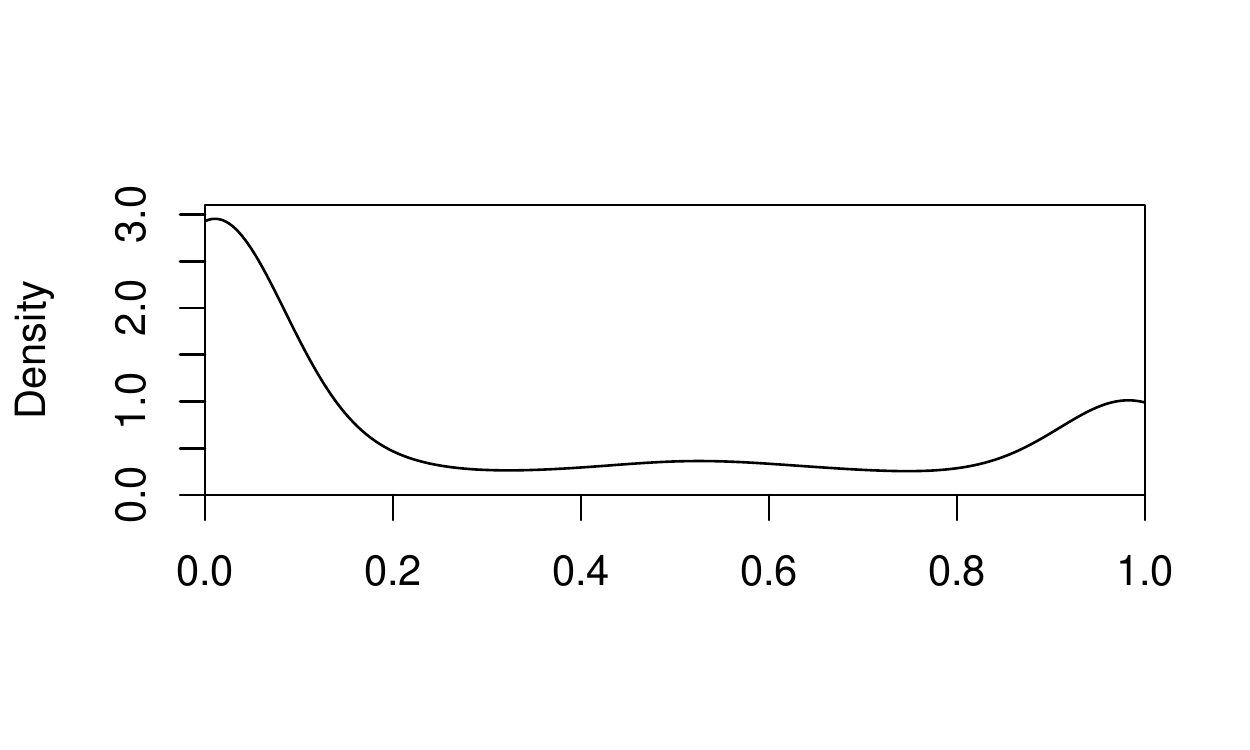}
\label{fig:prob.delta.a}}

\subfigure[Density for $10^4$ random
simulations.]{\includegraphics[width=0.99\textwidth]{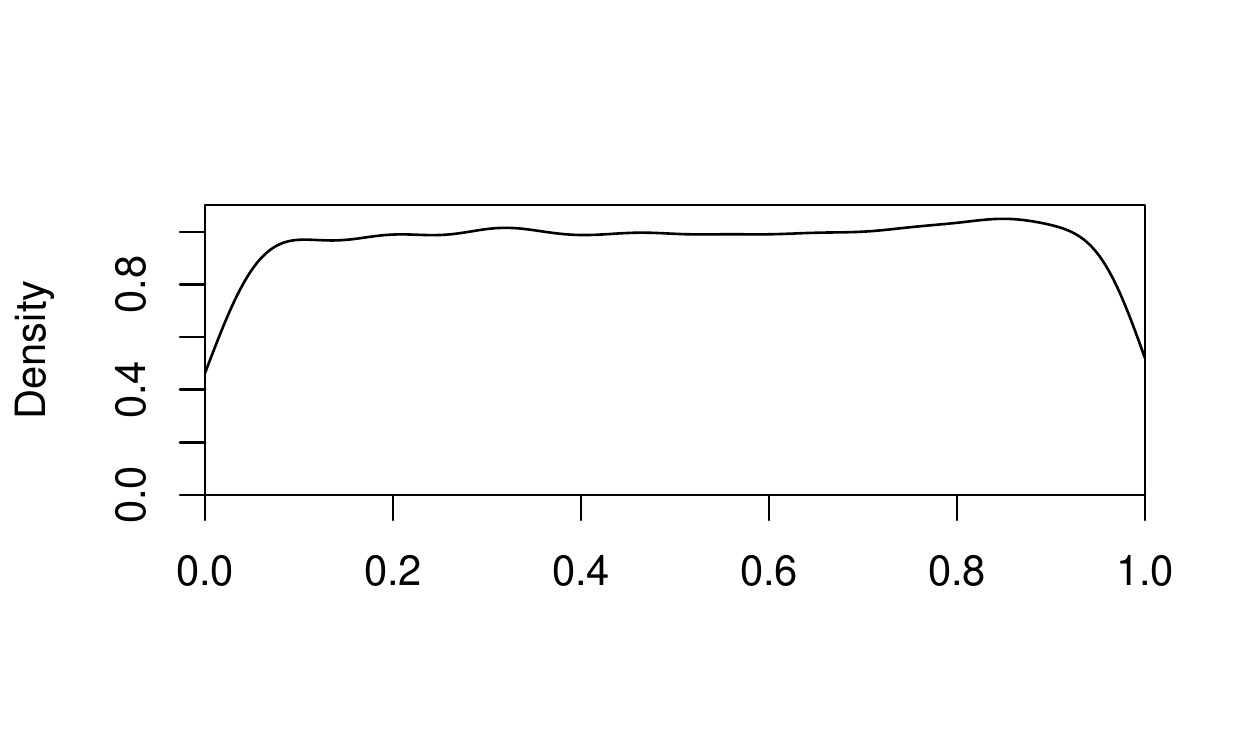}
\label{fig:prob.delta.b}}
\caption{Distributions of probabilities of observing a value smaller 
than $\Delta$ in real and randomly simulated data.}
\label{fig:prob.delta}
\end{figure}

The algorithm we used for uniformly sampling codon distributions
that extend a specified amino acid distribution can be described as follows. 
First fix the amino acid distribution $q=(q_A: A\in\alpha)$.  Then, a
distribution $p_{\cdot|A}=(p_{i|A}: i\in A)$ for the usage of the
synonymous codons that code for each amino acid~$A$ is sampled from the
appropriate Dirichlet distribution. As an example, 
a distribution for the two codons
(GAC and GAT) that code for aspartic acid (Asp), is obtained by sampling
from a Dirichlet$(1,1)$ distribution.
The final codon distribution $r=(r_i: i\in I)$ is then obtained by combining
these two distributions as follows:
$$
\forall i\in I, \quad r_i = p_{i|A^{i,\alpha}}q_{A^{i,\alpha}}.
$$
In other words, take an amino acid distribution and uniformly select a
possible SCU distribution given each amino acid in turn.  After
multiplying the SCU distribution for each amino acid by the probability of
the corresponding amino acid, the result is a codon distribution which is
a discrete extension of the original 
amino acid composition, making it possible
to compute $\Delta^q(\D,r)$ for the codon distribution.  the probability
$\Delta^*$ was then obtained as the fraction of the $10^4$ simulated codon
distributions whose~$\Delta$ is smaller than or equal to that of the
chromosome.

\medskip

\figref{fig:prob.delta.a} displays a kernel density 
estimate of the distribution
of~$\Delta^*$ for the collection of bacterial chromosomes.  Similarly,
\figref{fig:prob.delta.b} shows the distribution of~$\Delta^*$ for
uniformly sampled codon distributions. Three concentrations of
chromosomes are apparent in (a), the central one being fairly amorphous 
while the extreme concentrations stand out prominently. In contrast, 
$\Delta^*$ is essentially uniformly distributed in~(b).
In addition to occupying a lower 
dimensional subspace of the set of all 
codon distributions, there would seem to be something
particular about the way codon distributions 
for bacterial chromosomes are arranged.

\medskip

\begin{table}[ht]
\caption{Clustering of 2535 bacterial chromosomes into 
three groups according to 
$\Delta^*$, the proportion of codon distributions with less CUB than the
chromosome. The grouping was obtained via k-means clustering and 
the between-cluster sum of
squares accounts for $96.8\%$ of the total sum of squares.}
\label{tab:clusters}
\centering
\begin{tabular}{rlr|rrr|rrr}
  \hline
 Group & Bias & Size & \multicolumn3c{$\Delta^*$} & \multicolumn3c{$\Delta$}
 \\
& & & Center & Min. & Max. & Center & Min. & Max. \\ 
  \hline
1 & low & 1587 & 0.0269 & 0.0000 & 0.2661 & 0.1182 & 0.0262 & 0.2140 \\ 
  2 & moderate & 356 & 0.5086 & 0.2699 & 0.7240 & 0.2372 & 0.2042 & 0.2709 \\ 
  3 & high & 592 & 0.9556 & 0.7362 & 1.0000 & 0.3921 & 0.2645 & 0.6629 \\ 
   \hline
\end{tabular}
\end{table}
					
We used $k$-means clustering to assign each chromosome to one of three
groups according to the value of $\Delta^*$.  
The main characteristics of the resulting groups,
which we nominally named `low CUB', `moderate CUB'and `high CUB', are
summarized in Table~\ref{tab:clusters}.  The table gives the size 
of each group, the range of $\Delta^*$ and 
$\Delta$ spanned by each group, as well as the center (mean)
value of~$\Delta^*$ and~$\Delta$.  
Observe that the ranges of $\Delta^*$
covered by the three groups are disjoint, which is a side effect of the
$k$-means clustering procedure.  On the other hand, the ranges 
of~$\Delta$ that the groups encompass overlap slightly. 
This phenomenon is 
attributable to the theoretical equal weighting applied to all codon
distributions during the calculation of~$\Delta^*$.

\medskip

\begin{figure}[ht]
\centering
\includegraphics[width=\textwidth]{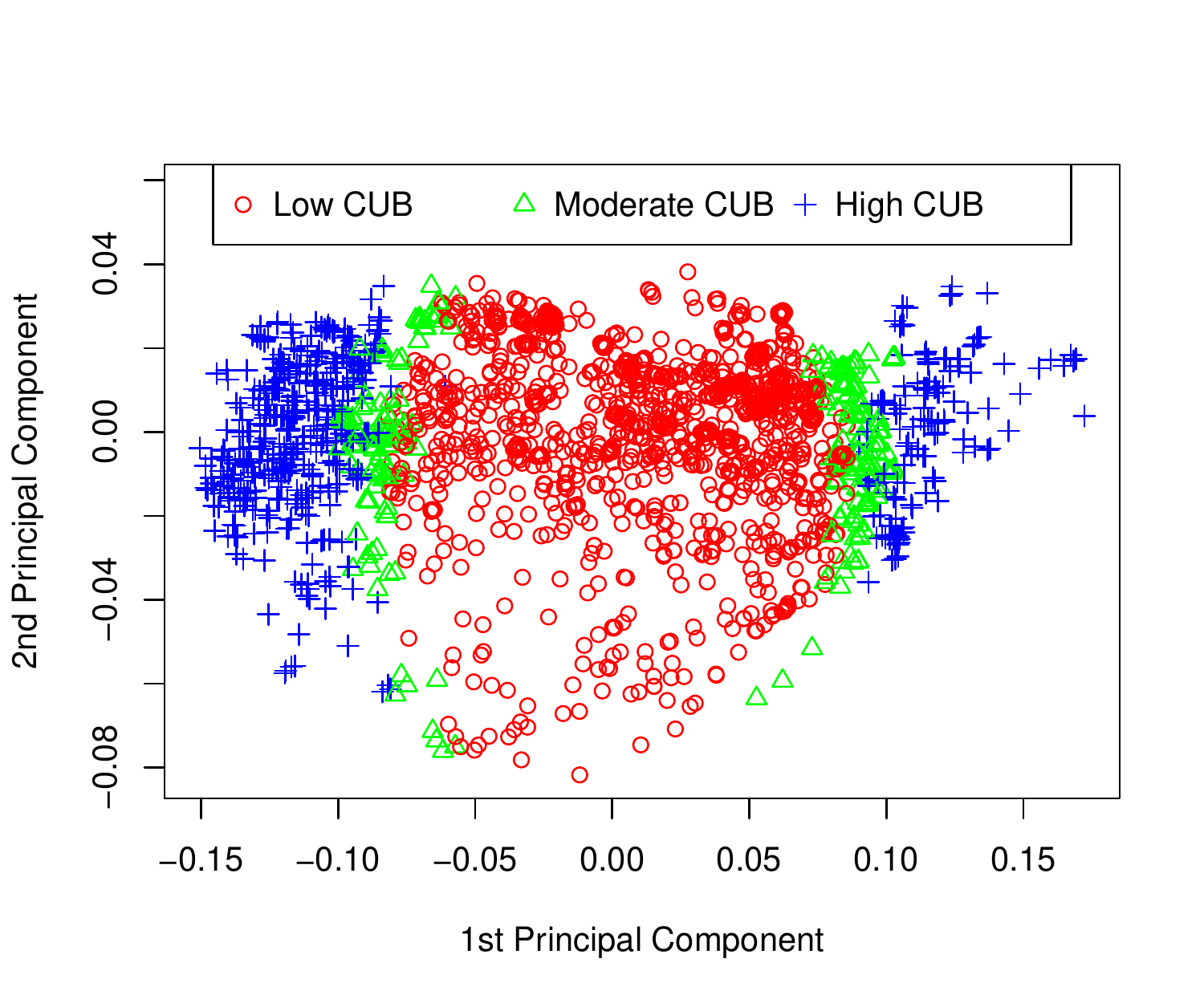}
\caption{Plot of the first two principal components of the codon 
distributions of 2535 bacterial chromosomes.}
\label{fig:pca}
\end{figure}

Next, we plotted the first two principal components of the codon 
distributions, indicating to which group each chromosome belongs. The
most significant feature of this plot (see \figref{fig:pca}) is that 
the plot is roughly broken into vertical bands according to membership 
in the `low CUB', ~moderate CUB'or `high CUB' group. Thus, the 
first principal component captures substantial information concerning 
CUB in the chromosome. We believe that this may be the first 
time that this characteristic of the first principal component 
has been demonstrated at the chromosome level.
At present, we have no satisfactory biological explanation for the 
division of bacterial chromosomes into three groups based on CUB.

\section*{Acknowledgments}

This work was supported by the Center for 
Mathematical Modeling's (CMM) CONICYT 
Basal program PFB 03. The authors 
would like to thank members of the CMM Mathomics Laboratory for
interesting and helpful discussions.


\end{document}